\documentclass[journal]{IEEEtran}
\usepackage{cite}
\usepackage[dvips]{graphicx}        
  \DeclareGraphicsExtensions{.eps}	
\usepackage{amssymb}
\usepackage[cmex10]{amsmath}
\usepackage{hyperref}
\hypersetup{letterpaper, pdfstartview=FitH, pdftitle={Scheduling for Optimal Rate Allocation in Ad Hoc Networks Subject to Heterogeneous Delay Constraints}, pdfauthor={Juan Jos\'e Jaramillo and R. Srikant}, bookmarksnumbered}
\newtheorem{theorem}{Theorem}
\newtheorem{lemma}{Lemma}
\newtheorem{corollary}{Corollary}
\newtheorem{fact}{Fact}
\newtheorem{definition}{Definition}

\begin{document}
\title{Scheduling for Optimal Rate Allocation in Ad Hoc Networks With Heterogeneous Delay Constraints}
\author{
	Juan Jos\'e Jaramillo,~\IEEEmembership{Member,~IEEE,} R. Srikant,~\IEEEmembership{Fellow,~IEEE,} and Lei Ying,~\IEEEmembership{Member,~IEEE}
	\thanks{Research supported by NSF Grants 07-21286, 05-19691, 03-25673, 08-31756, 09-53165, ARO MURI Subcontracts, AFOSR Grant FA-9550-08-1-0432, and DTRA Grants HDTRA1-08-1-0016, HDTRA1-09-1-0055.}
	\thanks{J. J. Jaramillo and L. Ying are with the Department of Electrical and Computer Engineering, Iowa State University, Ames, IA 50011 USA. Email: \{jjjarami, leiying\}@iastate.edu}
	\thanks{R. Srikant is with the Department of Electrical and Computer Engineering, and the Coordinated Science Laboratory, University of Illinois, Urbana, IL 61801 USA. Email: rsrikant@illinois.edu}
}
\maketitle

%
%
\begin{abstract}

This paper studies the problem of scheduling in single-hop wireless networks with real-time traffic, where every packet arrival has an associated deadline and a minimum fraction of packets must be transmitted before the end of the deadline. Using optimization and stochastic network theory we propose a framework to model the quality of service (QoS) requirements under delay constraints. The model allows for fairly general arrival models with heterogeneous constraints. The framework results in an optimal scheduling algorithm which fairly allocates data rates to all flows while meeting long-term delay demands. We also prove that under a simplified scenario our solution translates into a greedy strategy that makes optimal decisions with low complexity.

\end{abstract}

%
%

%
%
\section{Introduction}

\IEEEPARstart{I}{n} this paper we study the problem of scheduling real-time traffic in ad hoc networks under maximum per-packet delay constraints. The problem of scheduling best-effort traffic, which is defined as traffic that does not have any kind of quality of service (QoS) requirements such as minimum bandwidth or maximum delay, has been extensively studied for the case of wireless networks. An optimization framework for resource allocation in wireless networks has been developed in \cite{Eryilmaz05, Lin04, Neely05, Stolyar05, Eryilmaz06, Chen06}, where a dual decomposition approach was used to derive various components of the resource allocation architecture such as scheduling, congestion control, routing, power control, etc. A striking feature of the solution is an alternative derivation of the maxweight algorithm proposed in \cite{Tassiulas92}. We refer the readers to \cite{Lin06, Georgiadis06} for a survey of these works.

Scheduling algorithms for packets with strict deadline requirements have been proposed in \cite{Shakkottai02, Raghunathan08, Dua07, Liu06}, but the solutions are only approximate. In \cite{Hou09a, Hou09b, Hou09c}, the problem of optimal admission control and scheduling for real-time traffic was addressed for access-point wireless networks in which only one link can transmit at any given time. Among the many contributions in these papers is a key modeling innovation whereby the network is studied in frames, where a frame is a contiguous set of time-slots of fixed duration. Packets with deadlines are assumed to arrive only at the beginning of a frame and have to be served before the end of the frame according to some specified deadlines.	

The problem of optimal congestion control and scheduling for general ad hoc networks and arrivals was studied in \cite{Jaramillo10}, using the modeling paradigm of frames proposed in \cite{Hou09a}. The model allows a common framework for handling both best-effort and real-time traffic simultaneously, but it can only handle homogeneous per-packet delay requirements.

In this paper, we further extend the results of \cite{Jaramillo10} for the case of heterogeneous delays and more general arrival models, and study the impact of acknowledgments on scheduling in a lossy channel.

The main contributions of this work are summarized as follows:

\begin{enumerate}
  \item We present an optimization formulation for the problem of scheduling real-time traffic under maximum per-packet delay constraints in wireless ad hoc networks. Unlike earlier work, the formulation allows for general arrival models with heterogeneous delays.
  \item Using duality theory and a decomposition approach, we  design an optimal scheduler that fairly allocates data rates to all links and ensures that a required fraction of each flow's packets are delivered before the prescribed deadline by appealing to connections between Lagrange multipliers and service deficits.
  \item We prove that the scheduler meets all the QoS requirements and that converges to the optimal solution.
  \item We then consider noisy channels where the transmitter does not have perfect channel state information and relies on feedback (acknowledgment) from the receiver to find out if a transmission was successful. We derive the structure of optimal scheduling algorithms in this case and show that in the special case of colocated networks, the scheduling algorithm reduces to a simple greedy algorithm. The last part recovers a key result in \cite{Hou09a} using a different approach.
\end{enumerate}

The paper is organized as follows. Section \ref{network_model} presents the network model we use in this work. The optimization formulation is presented in Section \ref{formulation_kc} for the simplest of the channel models we study, while the dual decomposition approach is developed in Section \ref{decomposition_kc}. The optimal scheduler and its convergence properties are presented in Section \ref{online_alg_kc}. Since we study three different channel models, in Sections \ref{unknown_channel_case_pf} and \ref{unknown_channel_case_ps} we highlight the differences between the simplest channel model and the other two, and show the relationship between feedback after every transmission and algorithm complexity. In Section \ref{simulations} we perform a simulation study to understand the rates that can be achieved under different channel models. Finally, in Section \ref{conclusions} we present the conclusions.

%
%
\section{Network Model}
\label{network_model}

In this section we present our model for a network composed of single-hop traffic flows, such that each packet has maximum delay requirements.

We represent the network using a directed graph $\mathcal{G}=( \mathcal{N},\mathcal{L} )$, where $\mathcal{N}$ is the set of nodes and $\mathcal{L}$ is the set of links, such that for any $n_1, n_2 \in \mathcal{N}$, if $(n_1, n_2) \in \mathcal{L}$ then node $n_1$ can communicate to node $n_2$. Links are numbered $1$ through $|\mathcal{L}|$, and by abusing notation, we will sometimes use $l \in \cal L$ to mean $l\in\{1,2,\ldots, |\mathcal{L}|\}$.

Time is assumed to be divided in slots, and a set of $T$ consecutive slots is called a \emph{frame}. Let $\mathcal{T} \stackrel{def}{=} \{ 1,\ldots,T \}$. We denote by $a = ( a_{lt} )_{l \in \mathcal{L}, t \in \mathcal{T}}$ the number of packet arrivals at a given frame for link $l$ at time slot $t$, and assume that we get to know $a$ at the beginning of the frame. Furthermore, assume that $a_l \stackrel{def}{=} \sum_{t \in \mathcal{T}} a_{lt} $ is a random variable with mean $\lambda_l$ and variance $\sigma^2_{l}$, such that $Pr(a_{l}=0)>0$ and $Pr(a_{l}=1)>0$. The last two assumptions are used to guarantee that the Markov chain to be defined later is both irreducible and aperiodic, but they can be replaced by similar assumptions. We further assume that arrivals are independent between different frames.

Define $\mathcal{T}^a_l \stackrel{def}{=} \{ t: t \in \mathcal{T} \mbox{ and } a_{lt}>0 \}$ to be the set of arrival times at link $l$. Let $\tau = ( \tau_{lt} )_{l \in \mathcal{L}, t \in \mathcal{T}^a_l }$ be the deadline associated with packet arrivals. That is, a packet that arrived at link $l$, time $t$, must be transmitted by the end of time slot $\tau_{lt}$. We assume that the deadlines are such that
\begin{equation*}
	\{ t_1, \ldots, \tau_{lt_1} \} \cap \{ t_2, \ldots, \tau_{lt_2} \} = \emptyset \mbox{ for all } t_1, t_2 \in \mathcal{T}^a_l
\end{equation*}
and
\begin{equation*}
	\tau_{lt} \leq T \mbox{ for all } l \in \mathcal{L}, t \in \mathcal{T}^a_l.
\end{equation*}
In other words, packets must be transmitted before the next set of arrivals occurs in subsequent time slots, and all packets must be transmitted before the end of the frame.

If a packet misses its deadline it is discarded, and it is required that the loss probability at link $l \in \mathcal{L}$ due to deadline expiry must be no more than $p_l$. To avoid unnecessary complexity in the formulation, we will write $a$ to denote both the number of packet arrivals and its associated deadlines $\tau$.

This paper studies the cases where the channel state is assumed to be constant for the duration of a frame and when it is allowed to change from time slot to time slot. In both cases we assume the state is independent between frames and independent of arrivals. When the channel is fixed in a frame, let $c=( c_l )_{l \in \mathcal{L}}$ denote the number of packets link $l$ can successfully transmit in a time slot. When the channel is allowed to change, define $c=( c_{lt} )_{l \in \mathcal{L}, t \in \mathcal{T}}$ to be the number of packets that can be successfully transmitted at link $l$ in time slot $t$.

If we get to know the channel state before transmission, we can determine the optimal rate at which we can successfully transmit, possibly allowing more than one packet to be transmitted in a single slot. On the other hand, if the channel state is not known, we can only determine whether a transmission was successful or not after we get some feedback from the receiver. In this paper we try to capture these different scenarios in the following cases:
\begin{enumerate}
	\item Known channel state: It is assumed that $c_l$ is a \emph{non-negative} random variable with mean $\bar{c}_l$ and variance $\sigma^2_{cl}$, and we get to know the channel state at the \emph{beginning} of the frame.
	\item Unknown channel state, per-frame feedback: It is assumed that $c_l$ is a \emph{Bernoulli} random variable with mean $\bar{c}_l$ and we only get to know the channel state at the \emph{end} of the frame. In other words, the receiver acknowledges receptions for the all the packets in the frame, at the end of the frame.
	\item Unknown channel state, per-slot feedback: It is assumed that $c_{lt}$ is a \emph{Bernoulli} random variable with mean $\bar{c}_l$ and we get to know the channel state at the \emph{end} of the time slot. In other words, acknowledgments are received after each transmission.
\end{enumerate}

In the known channel state case we can potentially send more than one packet in a time slot at higher rates since channel estimation allow us to determine the optimal transmission rate. This is the reason why we make no assumptions on the values $c_l$ can take since it will be determined by the particular wireless technology used. In the case where the channel is unknown before transmission we assume that we only get binary feedback in the form of acknowledgments, which is reflected in the Bernoulli assumption on $c_l$ and $c_{lt}$. Thus, in this case, without any loss of generality, we assume only one packet can be transmitted per time slot per link.

For the sake of simplicity in the presentation, we will first develop the known channel case and we will later highlight the differences between this case and the other two cases in Sections \ref{unknown_channel_case_pf} and \ref{unknown_channel_case_ps}.

%
%
\section{Static Problem Formulation}
\label{formulation_kc}

To design our algorithm, we will first formulate the problem as a static optimization problem. Using duality theory, we will then obtain a dynamic solution to this problem and later we will prove its properties using stochastic Lyapunov techniques.

Let $s = ( s_{lt} )_{l \in \mathcal{L}, t \in \mathcal{T} }$ denote the number of packets scheduled for transmission at link $l$ and time slot $t$. We will only focus on feasible schedules, so if $s_{l_1t} > 0$ and $s_{l_2t} > 0$ for any $t$, then links $l_1$ and $l_2$ can be scheduled to simultaneously transmit without interfering with each other.

Since we cannot transmit more packets than what are available and what the channel state allows, we have the following constraints when the arrivals are given by $a$ and the channel state is $c$:
\begin{equation}
\label{arrival_constraint1_kc}
	\sum_{ j=t }^{ \tau_{lt} } s_{lj} \leq a_{lt} \mbox{ for all } t \in \mathcal{T}^a_l \mbox{, } l \in \mathcal{L} \mbox{,}
\end{equation}
\begin{equation}
\label{arrival_constraint2_kc}
	s_{lt} = 0 \mbox{ for all } t \in \mathcal{T} \setminus\cup_{ t \in \mathcal{T}^a_l } \{ t, \ldots, \tau_{lt} \} \mbox{, } l \in \mathcal{L} \mbox{, and }
\end{equation}
\begin{equation}
\label{channel_constraint_kc}
	s_{lt} \leq c_{l} \mbox{ for all } l \in \mathcal{L} \mbox{ and } t \in \mathcal{T}.
\end{equation}

Denote the set of feasible schedules when the arrivals and channel state are $a$ and $c$ by $\mathcal{S}(a, c)$, capturing any interference constraints imposed by the network and satisfying \eqref{arrival_constraint1_kc}, \eqref{arrival_constraint2_kc}, and \eqref{channel_constraint_kc}.

Our goal is to find $Pr( s | a, c)$ which is the probability of using schedule $s \in \mathcal{S}(a, c)$ when the arrivals are given by $a$ and the channel state is $c$, subject to the constraint that the loss probability at link $l \in \mathcal{L}$ due to deadline expiry cannot exceed $p_l$.

Denoting by $\mu(a,c) = ( \mu_{l}(a,c) )_{ l \in \mathcal{L} }$ the expected number of packets served when the arrivals and channel state are given by $a$ and $c$, respectively, we have:
\begin{equation*}
	\mu_{l}(a,c) \leq \sum\limits_{s \in \mathcal{S}(a, c)} \sum_{ t \in \mathcal{T} } s_{lt} Pr(s | a, c) \mbox{ for all } l \in \mathcal{L}.
\end{equation*}
Thus, the expected service at link $l \in \mathcal{L}$ is given by
\begin{equation*}
	\mu_{l} \stackrel{def}{=} \sum\limits_{a,c} \mu_{l}(a,c) Pr(c) Pr(a).
\end{equation*}
Due to QoS constraints we need at all links
\begin{equation*}
	\mu_{l} \geq \lambda_l(1-p_l),
\end{equation*}
and to avoid trivialities, we assume that $\lambda_l(1-p_l) > 0$ for all $l \in \mathcal{L}$.

For notational simplicity, define the capacity region for fixed arrivals and channel state as
\begin{equation*}
	\mathcal{C}(a,c) \stackrel{def}{=}
	\left\{
	\begin{array}{l}
		( \bar{\mu}_{l} )_{l \in \mathcal{L}}: \mbox{there exists } \bar{s} \in \mathcal{S}(a,c)_\mathcal{CH} \mbox{,} \\
		\bar{\mu}_{l} \leq \sum_{t \in \mathcal{T} } \bar{s}_{l,t}
	\end{array}
	\right\},
\end{equation*}
where $\mathcal{S}(a,c)_\mathcal{CH}$ is the convex hull of $\mathcal{S}(a,c)$.

Thus, the overall capacity region can be defined as follows:
\begin{equation*}
	\mathcal{C} \stackrel{def}{=}
	\left\{
	\begin{array}{l}
		( \mu_{l} )_{l \in \mathcal{L}} : \mbox{there exists } ( \bar{\mu}_{l}(a,c) )_{l \in \mathcal{L}} \in \mathcal{C}(a,c) \\
		\mbox{ for all } a, c \mbox{ and } \mu_{l} = E[ \bar{\mu}_{l}(a,c) ] \mbox{ for all } l \in \mathcal{L} \\
	\end{array}
	\right\}.
\end{equation*}

We will focus on the following static formulation for our problem, for some given vector $w \in \mathbb{R}_+^{|\mathcal{L}|}$:
\begin{equation}
\label{offline_opt_kc}
	 \max\limits_{ \mu \in \mathcal{C} } \sum\limits_{l \in \mathcal{L}} w_l \mu_{l}
\end{equation}
subject to
\begin{equation*}
	\mu_{l} \geq \lambda_{l}(1-p_l) \mbox{ for all } l \in \mathcal{L}.
\end{equation*}

The vector $w$ can be used to allocate additional bandwidth fairly to flows beyond what is required to meet their QoS needs. Other uses for $w$ have been explored in \cite{Jaramillo10}. We will assume that the arrivals and loss probability requirements are feasible and thus the optimization problem has a solution $\mu^*$.

%
%
\section{Dual Decomposition of the Static Problem}
\label{decomposition_kc}

In this section we use duality theory to decompose the static optimization problem into simpler subproblems that will give us the ideas behind the dynamic algorithm.

Using the definition of the dual function\cite{Luenberger03}, we have that
\begin{equation*}
	D(\delta) = \max\limits_{ \mu \in \mathcal{C} } \sum\limits_{l \in \mathcal{L}} w_l \mu_{l} - \delta_{l} [\lambda_l(1-p_l) - \mu_{l}].
\end{equation*}
From Slater's condition \cite{Boyd04} we know that the duality gap is zero and therefore $D(\delta^*) = \sum\limits_{l \in \mathcal{L}} w_l \mu_{l}^*$, where
\begin{equation*}
	\delta^* \in \mathop{\arg\min}_{\delta_{l} \geq 0} D(\delta).
\end{equation*}

We are interested in finding $\mu^*$ but not the value $D(\delta^*)$, so the problem can be simplified as follows
\begin{equation}
\label{first_decomposition_kc}
	\max\limits_{ \mu \in \mathcal{C} } \sum\limits_{l \in \mathcal{L}} ( w_l + \delta_{l} ) \mu_{l}.
\end{equation}
Since we are interested in solving the problem for non-negative values of $\delta_{l}$, it must be the case that $\mu^*$ is as large as the constraints allow. Furthermore, since the objective function in (\ref{first_decomposition_kc}) is linear, the problem can be decomposed into the following subproblems for fixed $a$ and $c$:
\begin{equation*}
	\max\limits_{s \in \mathcal{S}(a,c)} \sum\limits_{l \in \mathcal{L}} (w_l+\delta_{l}) \sum_{ t\in \mathcal{T} } s_{lt}.
\end{equation*}

The optimization problem can be solved using the following iterative algorithm, where $k$ is the step index:
\begin{align*}
	\tilde{s}^*(a,c,k) \in \mathop{\arg\max}\limits_{s \in \mathcal{S}(a,c)} \sum\limits_{l \in \mathcal{L}} [w_l+\delta_{l}(k)] \sum_{ t\in \mathcal{T} } s_{lt}
\end{align*}
\begin{equation*}
	\tilde{\mu}_{l}^*(k) = \sum\limits_{ a,c } \sum_{ t\in \mathcal{T} } \tilde{s}_{lt}^*(a, c, k) Pr(c) Pr(a).
\end{equation*}
And the update equation for the Lagrange multipliers is given by
\begin{equation*}
	\delta_{l}(k+1) = \{ \delta_{l}(k) + \epsilon [ \lambda_l (1-p_l) - \tilde{\mu}_{l}^*(k) ] \}^+,
\end{equation*}
where $\epsilon>0$ is a fixed step-size parameter, and for any $\alpha \in \mathbb{R}$, $\alpha^+ \stackrel{def}{=}\max\{ \alpha, 0 \}$.

With the change of variables $\epsilon \hat{d}(k) = \delta(k)$, we rewrite the algorithm as
\begin{equation}
\label{decomposition_algorithm_kc}
	\tilde{s}^*(a,c,k) \in \mathop{\arg\max}\limits_{s \in \mathcal{S}(a,c)} \sum\limits_{l \in \mathcal{L}} [ \frac{1}{\epsilon} w_l+\hat{d}_{l}(k) ] \sum_{ t\in \mathcal{T} } s_{lt}
\end{equation}
\begin{equation*}
	\tilde{\mu}_{l}^*(k) = \sum\limits_{ a,c } \sum_{ t\in \mathcal{T} } \tilde{s}_{lt}^*(a, c, k) Pr(c) Pr(a),
\end{equation*}
with update equation:
\begin{equation*}
	\hat{d}_{l}(k+1) = [ \hat{d}_{l}(k) +\lambda_l (1-p_l) - \tilde{\mu}_{l}^*(k) ] ]^+.
\end{equation*}
From the update equation we see that $\hat{d}_l(k)$ can be interpreted as a queue with $\lambda_l(1-p_l)$ arrivals and $\tilde{\mu}_{l}^*(k)$ departures at step $k$.

%
%
\section{Online Algorithm and Its Convergence Analysis}
\label{online_alg_kc}

So far we have presented a dual decomposition for a static problem; however, the real network has stochastic arrivals and channel state. We will use the intuition from the decomposition in this section to develop an online algorithm that can cope with such changing conditions and prove its convergence properties.

\subsection{Scheduler}
\label{online_kc}

We propose the following dynamic scheduling algorithm, where the arrivals and channel state  in frame $k$ are given by $a(k)$ and $c(k)$, respectively:
\begin{equation}
\label{online_opt_sch_kc}
	\tilde{s}^*( a(k),c(k),d(k) ) \in \mathop{\arg\max}\limits_{s \in \mathcal{S}( a(k),c(k) )} \sum\limits_{l \in \mathcal{L}} [ \frac{1}{\epsilon} w_l+d_{l}(k) ] \sum_{ t \in \mathcal{T} } s_{lt},
\end{equation}
with update equation
\begin{equation*}
	d_{l}(k+1) = [ d_{l}(k) + \tilde{a}_{l}(k) - I_{l}^*( a(k),c(k),d(k) ) ]^+,
\end{equation*}
where
\begin{equation*}
	I_{l}^*( a(k),c(k),d(k) ) = \sum_{ t\in \mathcal{T} } \tilde{s}_{lt}^*( a(k),c(k),d(k) )
\end{equation*}
and $\tilde{a}_{l}(k)$ is a binomial random variable with parameters $a_{l}(k)$ and $1-p_l$. The quantity $\tilde{a}_{l}(k)$ can be generated by the network as follows: upon each packet arrival, toss a coin with probability of \emph{heads} equal to $1-p_l$, and if the outcome is \emph{heads}, add a one to the deficit counter $d_l(k)$. This implementation for $\tilde{a}_{l}(k)$ was first suggested in \cite{Jaramillo10}.

Note that in \eqref{online_opt_sch_kc} we make explicit the fact that the optimal scheduler is a function of  $a(k)$, $c(k)$, and $d(k)$, for fixed $w$ and $\epsilon$. Also note that $d_l(k)$ can be interpreted as a virtual queue that keeps track of the deficit in service for link $l$ to achieve a loss probability due to deadline expiry less than or equal to $p_l$. The idea of using a deficit counter was first used in \cite{Hou09a} for the case of colocated networks with homogeneous delays, while the Lagrange multiplier interpretation allowed us to extend the result to general ad hoc networks and heterogeneous delays.

\subsection{Convergence Results}
\label{convergence_kc}

We now prove that the online algorithm meets the QoS constraints, the total expected service deficit has a $O(1 / \epsilon)$ bound, and the expected value of the objective is within $O(\epsilon)$ of the optimal value of the static problem (\ref{offline_opt_kc}).
For the sake of readability, we defer the proofs to the appendixes.

We will first bound the expected drift of $d(k)$ for a suitable Lyapunov function. Note that $d(k)$ defines an irreducible and aperiodic Markov chain.

\begin{lemma}
\label{expected_drift_kc}
Consider the Lyapunov function $V(d)=\frac{1}{2}\sum_{l \in \mathcal{L}}d_l^2$. If there exists a point $\mu(\Delta) \in \mathcal{C}/(1+\Delta)$ for some $\Delta > 0$ such that
\begin{equation}
\label{inelastic_feasibility}
	\mu_{l}(\Delta) \geq \lambda_{l}(1-p_l) \mbox{ for all } l \in \mathcal{L}
\end{equation}
then
\begin{align*}
	E & \left[ V(d(k+1)) | d(k)=d \right] - V(d) \leq B_1 - B_2 \sum_{l \in \mathcal{L}} d_l \\
 	& - \frac{1}{\epsilon} \sum_{l \in \mathcal{L}} \left\{ w_l (1+\Delta) \mu_{l}(\Delta) - w_l E \left[ I_{l}^*(a(k),c(k),d) \right] \right\} \\
\end{align*}
for some positive constants $B_1$, $B_2$, any $\epsilon>0$, where $I^*(a(k),c(k),d)$ is obtained from the solution to (\ref{online_opt_sch_kc}).
$\hfill \diamond$
\end{lemma}

Since the last term in the right-hand side of the inequality can be upper-bounded, it can be shown that the expected drift is negative but for a finite set of values of $d(k)$. Thus, Lemma \ref{expected_drift_kc}  implies that $d(k)$ is positive recurrent. As a corollary of this result, we have that the total expected service deficit has an $O(1 / \epsilon)$ bound.

\begin{corollary}
\label{def_queue_bound_kc}
If there exists a point $\mu(\Delta) \in \mathcal{C}/(1+\Delta)$ for some $\Delta > 0$ such that (\ref{inelastic_feasibility}) holds true, then the total expected service deficit is upper-bounded by
\begin{equation*}
	\limsup_{k \rightarrow \infty} E\left[ \sum_{l \in \mathcal{L}} d_l(k) \right] \leq B_3 + \frac{1}{\epsilon} B_4
\end{equation*}
for all $l \in \mathcal{L}$, where $B_3 = B_1/B_2 \mbox{ and } B_4 = \sum_{l \in \mathcal{L}} w_l \lambda_{l} / B_2. \hfill \diamond$
\end{corollary}

Lemma \ref{expected_drift_kc} also implies that the scheduler fulfills all QoS requirements.
\begin{corollary}
If there exists a point $\mu(\Delta) \in \mathcal{C}/(1+\Delta)$ for some $\Delta > 0$ such that (\ref{inelastic_feasibility}) holds true, then the online algorithm fulfills all the QoS constraints. That is:
\begin{equation*}
	\liminf_{K \rightarrow \infty} E \left[ \frac{1}{K} \sum_{k=1}^{K} I_{l}^*(a(k),c(k),d(k)) \right] \geq \lambda_{l}(1-p_l)
\end{equation*}
for all $l \in \mathcal{L}$.
$\hfill \diamond$
\end{corollary}
The above corollary states that the arrival rate into the deficit counter is less than or equal to the service rate. The result is an obvious consequence of the stability of the deficit counters and so a formal proof is not given.

In order to be able to prove that our scheduler achieves the optimal solution to the static problem in an average sense, a related result to Lemma \ref{expected_drift_kc} must be stated first. Since the proof is similar to the proof in Lemma \ref{expected_drift_kc}, it is omitted.

\begin{lemma}
\label{bound_expected_drift_kc}
Consider the Lyapunov function $V(d)=\frac{1}{2}\sum_{l \in \mathcal{L}}d_l^2$. Then
\begin{align*}
	E & \left[ V(d(k+1)) | d(k)=d \right] - V(d) \leq B_1 - B_5 \sum_{l \in \mathcal{L}} d_l \\
 	& - \frac{1}{\epsilon} \sum_{l \in \mathcal{L}} \left\{ w_l \mu^*_{l} - w_l E \left[ I_{l}^*(a(k),c(k),d) \right] \right\}
\end{align*}
for $B_1 > 0$, some non-negative constant $B_5$, any $\epsilon>0$, where $\mu^*$ is a solution to \eqref{offline_opt_kc}, and $I^*(a(k),c(k),d)$ is obtained from the solution to (\ref{online_opt_sch_kc}).
$\hfill \diamond$
\end{lemma}

Using Lemma \ref{bound_expected_drift_kc} we can now prove that our online algorithm is within $O(\epsilon)$ of the optimal value.
\begin{theorem}
\label{optimality_online_kc}
For any $\epsilon >0$ we have that
\begin{align*}
	\limsup_{K \rightarrow \infty} & E \left[ \sum_{l \in \mathcal{L}} w_l \mu^*_{l} - w_l \frac{1}{K} \sum_{k=1}^K I_{l}^*(a(k),c(k),d(k)) \right] \leq B \epsilon
\end{align*}
for some $B>0$, where $\mu^*$ is a solution to (\ref{offline_opt_kc}), and $I^*(a(k),c(k),d(k))$ is obtained from the solution to (\ref{online_opt_sch_kc}).
$\hfill \diamond$
\end{theorem}

Both Corollary \ref{def_queue_bound_kc} and Theorem \ref{optimality_online_kc} highlight the trade-off in choosing $\epsilon$: the closer we get to the optimal value for the static formulation, the more the aggregate in the deficit counters will increase.

The statement and proofs of Lemma \ref{expected_drift_kc} and Theorem \ref{optimality_online_kc} follow the techniques in \cite{Jaramillo10}, which are similar to the techniques in \cite{Neely05}. Slightly different results can be derived using the techniques in \cite{Stolyar05} and \cite{Eryilmaz05}.

%
%
\section{Unknown Channel State, Per-Frame Feedback}
\label{unknown_channel_case_pf}

Since the analysis for the unknown channel state, per-frame feedback case is similar in nature to the known channel case, in this section we will only highlight the differences between both cases.

Remember that we assume that the channel at link $l \in \mathcal{L}$, $c_l$, is a \emph{Bernoulli} random variable with mean $\bar{c}_l$, and we get to know the channel state only at the end of the frame. Thus, the set of feasible schedules when the arrivals are given by $a$ is the set of schedules that capture any interference constraints imposed by the network and that fulfill the following constraints:
\begin{equation}
\label{arrival_constraint1_uc}
	\sum_{ j=t }^{ \tau_{lt} } s_{lj} \leq a_{lt} \mbox{ for all } t \in \mathcal{T}^a_l \mbox{, } l \in \mathcal{L} \mbox{,}
\end{equation}
\begin{equation}
\label{arrival_constraint2_uc}
	s_{lt} = 0 \mbox{ for all } t \in \mathcal{T} \setminus\cup_{ t \in \mathcal{T}^a_l } \{ t, \ldots, \tau_{lt} \} \mbox{, } l \in \mathcal{L} \mbox{, and }
\end{equation}
\begin{equation}
\label{channel_constraint_uc}
	s_{lt} \leq 1 \mbox{ for all } l \in \mathcal{L} \mbox{ and } t \in \mathcal{T}.
\end{equation}

From our assumption that the channel state remains constant for the duration of a frame, \eqref{arrival_constraint1_uc} and \eqref{arrival_constraint2_uc} tell us that we should not schedule a link for transmission more than the number of packets that are available, since extra transmissions do not increase the number of successes. Furthermore, since the channel is Bernoulli, when a link is scheduled it can only transmit at most a packet, as the constraint in \eqref{channel_constraint_uc} indicates. Let $\mathcal{S}(a)$ denote the set of feasible schedules when the arrivals are given by $a$, capturing any interference constraints given by the network and fulfilling \eqref{arrival_constraint1_uc}, \eqref{arrival_constraint2_uc}, and \eqref{channel_constraint_uc}.

Following the same arguments as in Section \ref{formulation_kc}, we have to design a scheduling strategy $Pr(s | a)$ that is the probability of using schedule $s \in \mathcal{S}(a)$ when the arrivals are given by $a$. Since we cannot base our policy on the state of the channel, as we did before, the expected service to link $l \in \mathcal{L}$ is now given by
\begin{equation*}
	\mu_{l} \leq \sum\limits_{a,c} \sum\limits_{s \in \mathcal{S}(a)} \sum_{ t \in \mathcal{T} } c_l s_{lt} Pr(s | a) Pr(c) Pr(a).
\end{equation*}
Simplifying the equation we get
\begin{equation}
\label{new_def_uc}
	\mu_{l} \leq \sum\limits_{a} \sum\limits_{s \in \mathcal{S}(a)} \sum_{ t \in \mathcal{T} } \bar{c}_l s_{lt} Pr(s | a) Pr(a).
\end{equation}
That is, the expected service depends on the average channel state. Properly modifying the definition of the capacity of the network, we can write the static formulation as in \eqref{offline_opt_kc}.

Using the dual decomposition technique of Section \ref{decomposition_kc} we can develop the design ideas behind the following dynamic scheduler, assuming that at frame $k$ the arrivals are given by $a(k)$:
\begin{equation}
\label{online_opt_sch_uc}
	\tilde{s}^*( a(k),d(k) ) \in \mathop{\arg\max}\limits_{s \in \mathcal{S}( a(k) )} \sum\limits_{l \in \mathcal{L}} [ \frac{1}{\epsilon} w_l+d_{l}(k) ] \bar{c}_l \sum_{ t \in \mathcal{T} } s_{lt},
\end{equation}
with update equation
\begin{equation*}
	d_{l}(k+1) = [ d_{l}(k) + \tilde{a}_{l}(k) - I_{l}^*( a(k),d(k) ) ]^+,
\end{equation*}
where
\begin{equation*}
	I_{l}^*( a(k),d(k) ) = \sum_{ t\in \mathcal{T} } \tilde{s}_{lt}^*( a(k),d(k) )
\end{equation*}
and $\tilde{a}_{l}(k)$ is a binomial random variable with parameters $a_{l}(k)$ and $1-p_l$.

The main difference in the scheduler compared to the known channel case is that the algorithm now uses the expected channel state in making scheduling decisions. Thus, the algorithm needs to know or estimate $\bar{c}_l$.

Using the same techniques as in Section \ref{convergence_kc}, we can prove that the scheduler meets all the QoS requirements, the total expected service deficits have an $O(1 / \epsilon)$ bound, and the mean value of the objective is within $O(\epsilon)$ of the optimal value.

%
%
\section{Unknown Channel State, Per-Slot Feedback}
\label{unknown_channel_case_ps}


Compared to the previous two cases, the per-slot feedback case is more complex due to the fact that we can use the feedback to update our decisions at every time slot. In this section we will first formulate the problem focusing on policies rather than on schedules, and we will show that no simple decomposition can be achieved for this case. However, we will prove that for a simple scenario a greedy solution can achieve the optimal solution.

\subsection{Problem Formulation and Solution}
\label{formulation_uc_ps}

Similar to the development in Section \ref{unknown_channel_case_pf}, we will only highlight the differences between this case and the known channel state case.

As described in Section \ref{network_model}, the channel at link $l \in \mathcal{L}$, time slot $t \in \mathcal{T}$, $c_{lt}$, is assumed to be a Bernoulli random variable with mean $\bar{c}_l$. Thus, instead of choosing a schedule for the entire frame, we will try to find a scheduling policy $\rho$ that makes decisions at every time slot based on the feedback received. Note however that if the arrivals and channel state at a given frame are given by $a$ and $c$ respectively, policy $\rho$ will generate a schedule by the end of the frame. We will denote by $s(\rho,a,c)$ such schedule.

We only focus on feasible policies, which are defined to be policies that generate a schedule that meets all interference constraints given by the network and fulfill the following constraints, for fixed $\rho$, $a$, and $c$:
\begin{equation}
\label{arrival_constraint2_ps}
	s_{lt}(\rho,a,c) = 0 \mbox{ for all } t \in \mathcal{T} \setminus\cup_{ t \in \mathcal{T}^a_l } \{ t, \ldots, \tau_{lt} \} \mbox{, } l \in \mathcal{L},
\end{equation}
\begin{equation}
\label{channel_constraint_ps}
	s_{lt}(\rho,a,c) \leq 1 \mbox{ for all } l \in \mathcal{L} \mbox{ and } t \in \mathcal{T} \mbox{, and }
\end{equation}
\begin{equation}
\label{arrival_constraint1_ps}
	\sum_{ j=t }^{ \tau_{lt} } c_{lj} s_{lj}(\rho,a,c) \leq a_{lt} \mbox{ for all } t \in \mathcal{T}^a_l \mbox{, } l \in \mathcal{L}.
\end{equation}

Note that \eqref{arrival_constraint2_ps} specifies that a link should not be scheduled if there is no packet to be transmitted, \eqref{channel_constraint_ps} states that we cannot schedule more than a packet in a time slot since the channel is Bernoulli, and \eqref{arrival_constraint1_ps} specifies that a feasible policy cannot have more successful transmissions than the number of packets available.

We highlight the fact that since there is only a finite number of feasible schedules, then the set of feasible policies is finite. We will denote by $\mathcal{P}(a)$ the set of feasible policies that meet all interference constraints and that fulfill \eqref{arrival_constraint2_ps}, \eqref{channel_constraint_ps}, and \eqref{arrival_constraint1_ps}, when arrivals are given by $a$.

Our goal is to find the probability distribution $Pr(\rho | a)$ of using policy $\rho \in \mathcal{P}(a)$ in a given frame when arrivals are given by $a$, such that the fraction of packets that miss the deadline at link $l$ cannot exceed $p_l$. Thus, the expected service at link $l \in \mathcal{L}$ is subject to the following constraint
\begin{equation*}
	\mu_{l} \leq \sum\limits_{a,c} \sum\limits_{\rho \in \mathcal{P}(a)} \sum_{ t \in \mathcal{T} } c_{lt} s_{lt}(\rho, a, c) Pr(\rho | a) Pr(c) Pr(a).
\end{equation*}

Therefore, the optimization problem is as follows, for a given vector $w \in \mathbb{R}_+^{|\mathcal{L}|}$:
\begin{equation}
\label{offline_opt_ps}
	 \max\limits_{ \mu, Pr(\rho | a) } \sum\limits_{l \in \mathcal{L}} w_l \mu_{l}
\end{equation}
subject to
\begin{equation*}
	\mu_{l} \leq \sum\limits_{a,c} \sum\limits_{\rho \in \mathcal{P}(a)} \sum_{ t \in \mathcal{T} } c_{lt} s_{lt}(\rho, a, c) Pr(\rho | a) Pr(c) Pr(a) \mbox{ for all } l
\end{equation*}
\begin{equation*}
	\mu_{l} \geq \lambda_{l}(1-p_l) \mbox{ for all } l \in \mathcal{L}
\end{equation*}
\begin{equation*}
	Pr(\rho | a) \geq 0 \mbox{ for all } \rho \in \mathcal{P}(a), a
\end{equation*}
\begin{equation*}
	\sum_{\rho \in \mathcal{P}(a)} Pr(\rho | a) \leq 1  \mbox{ for all } a.
\end{equation*}
We will assume that the arrivals and loss probability requirements are feasible and thus the optimization problem has a solution $\mu^*$.

Following the arguments in Section \ref{decomposition_kc}, we can develop the design ideas behind the following dynamic scheduler, assuming that at frame $k$ the arrivals are given by $a(k)$ and the channel state by $c(k)$:
\begin{equation}
\label{online_opt_sch_uc_ps}
	\tilde{\rho}^*(a(k),d(k)) \in \mathop{\arg\max}\limits_{\rho \in \mathcal{P}(a(k))} \sum\limits_{l \in \mathcal{L}} [ \frac{1}{\epsilon} w_l+d_{l}(k) ] \mu_{l}(\rho,a(k))
\end{equation}
with update equation
\begin{equation*}
	d_{l}(k+1) = [ d_{l}(k) + \tilde{a}_{l}(k) - I_{l}^*( a(k),c(k),d(k) ) ] ]^+,
\end{equation*}
where
\begin{equation}
	\mu_{l}(\rho,a(k)) = \sum\limits_{c} \sum_{ t \in \mathcal{T} } c_{lt} s_{lt}(\rho, a(k), c) Pr(c),
\end{equation}
\begin{equation*}
	I_{l}^*( a(k),c(k),d(k) ) = \sum_{ t \in \mathcal{T} } c_{lt}(k) s_{lt}(\tilde{\rho}^*(k), a(k), c(k)),
\end{equation*}
$\tilde{\rho}^*(k) = \tilde{\rho}^*(a(k),d(k))$, and $\tilde{a}_{l}(k)$ is a binomial random variable with parameters $a_{l}(k)$ and $1-p_l$.

We note that compared to the known channel case, the algorithm needs to know the probability distribution of $c$ in order to make optimal decisions. From \eqref{online_opt_sch_uc_ps} we see that the duality approach does not give us a simple decomposition as in \eqref{decomposition_algorithm_kc}. The reason comes from the fact that even though per-slot feedback may help us to potentially increase the throughput, it also increases the complexity of the decision algorithm.

Using the same proof techniques as in Section \ref{convergence_kc}, it can be proved that the scheduler meets all the QoS requirements, the total expected service deficits have an $O(1 / \epsilon)$ bound, and the mean value of the objective is within $O(\epsilon)$ of the optimal value.

\subsection{A Greedy Strategy for Colocated Networks}
\label{greedy_strategy_uc_ps}

In this Section we will show that in a simple scenario a greedy algorithm can achieve the optimal solution with minimum complexity. To do that, we will focus our attention to colocated networks, where only one link is allowed to transmit at any given time slot. We will also assume that the channel state is independent between different time slots. Furthermore, we will assume that at every frame there is a single packet arrival at every link at the beginning of the frame, and that all the packets must be transmitted by the end of the frame. That is,
\begin{equation}
\label{arrivals_gs_uc}
	\mathcal{T}^a_l = \{ 1 \} \mbox{, } a_{l1} = 1 \mbox{, and  } \tau_{l1} = T \mbox{ for all }l \in \mathcal{L}.
\end{equation}

The key idea we will use in this section is that for a given frame when the deficit counters are given by $d$, links will be prioritized in decreasing order of the priorities $[ \frac{1}{\epsilon} w_l + d_l ] \bar{c}_l$.

\begin{definition}
A \emph{greedy policy} for colocated networks is a scheduling policy that at every time slot schedules a link with the highest priority $[ \frac{1}{\epsilon} w_l + d_l ] \bar{c}_l$ among the links that have a packet that remains to be transmitted.
$\hfill \diamond$
\end{definition}

\begin{theorem}
\label{optimality_greedy}
The greedy scheduler is the optimal solution to \eqref{online_opt_sch_uc_ps} for colocated networks, when the arrivals are given by \eqref{arrivals_gs_uc}, and the channel state is independent between different time slots.
$\hfill \diamond$
\end{theorem}

For ease of readability, we defer the proof to the appendixes. Due to the optimality of the policy, and following a similar development as in Section \ref{convergence_kc}, one can prove that the greedy algorithm meets all the QoS constraints, the total expected service deficits have an $O(1 / \epsilon)$ bound, and the mean value of the objective is within $O(\epsilon)$ of the optimal value. We skip the proofs since they are analogous to the ones already presented.

The above theorem shows that the dual decomposition solution presented here recovers the solution for the special case of access-point networks presented in \cite{Hou09a}. The contribution of this section is to show that the dual approach allow us to extend such results for very general ad hoc networks, arrivals, and for heterogeneous delays, and that \cite{Hou09a} can be seen as a particular case of our general formulation.

%
%
\section{Simulations}
\label{simulations}

The purpose of the simulations is to compare the throughput that can be achieved for the different channel models. To do that, we simulate a 10-link network such that its interference graph is given by Fig. \ref{interference_graph}, where each vertex of the graph represents a link and the edges represent the interference constraints. For example, if link 1 is scheduled, then links 2, 4, and 7 cannot be activated. This interference graph was first used in \cite{Jaramillo10}.

\begin{figure}[t]
	\centering
	\includegraphics[angle=0, width=2.0in]{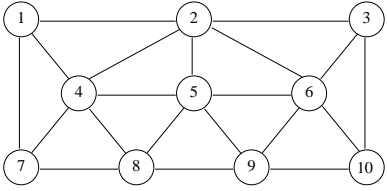}
	\caption{Interference graph used in the simulations}
	\label{interference_graph}
\end{figure}

The constraint for the dropping probability due to deadline expiration is set to 0.1, the packet arrivals at every link are assumed to be Bernoulli random variables with mean 0.6 packets/frame, and a frame has 3 time slots. Every channel is assumed to be a Bernoulli random variable with mean 0.96. We will compare the three channel models studied in this paper: known channel state, unknown channel state per-frame feedback, and unknown channel state per-slot feedback. The simulation time is $10^6$ frames.

In Figs. \ref{rate_ten_w0_ch} and \ref{rate_ten_w6_ch} we plot the average service for different values of $w_l$. The case $w_l=0$ means that we are only interested in finding a feasible solution without any concern of optimality. This results in a rate assignment that is only slightly above the minimum required to achieve an acceptable dropping probability, as shown in Fig. \ref{drop_prob_ten_w0_ch}. In the case $w_l=6$ we have that the algorithm tries to maximize the service rate for all links, resulting in a significant decrease of the dropping probability as shown in Fig. \ref{drop_prob_ten_w6_ch}. This result suggests that the objective function not only has a role giving priorities to links according to its weights, but also has a role in decreasing the dropping probabilities, allocating any available capacity to links.

It is interesting to note that the difference in service rates for the different channel models is smaller than $2\%$, which suggests that the added complexity in solving the optimal scheduler for the per-slot feedback case does not compare with the gains that can be achieved in terms of throughput. Thus, the per-frame feedback case can be used as a low-complexity approximation to the per-slot feedback case.

\begin{figure}[t]
	\centering
	\includegraphics[angle=270, width=3.3in]{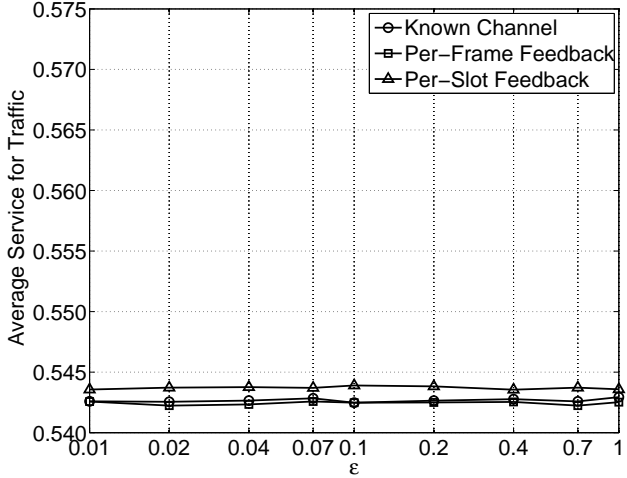}
	\caption{Average service when $w_l=0$}
	\label{rate_ten_w0_ch}
\end{figure}
\begin{figure}[t]
	\centering
	\includegraphics[angle=270, width=3.3in]{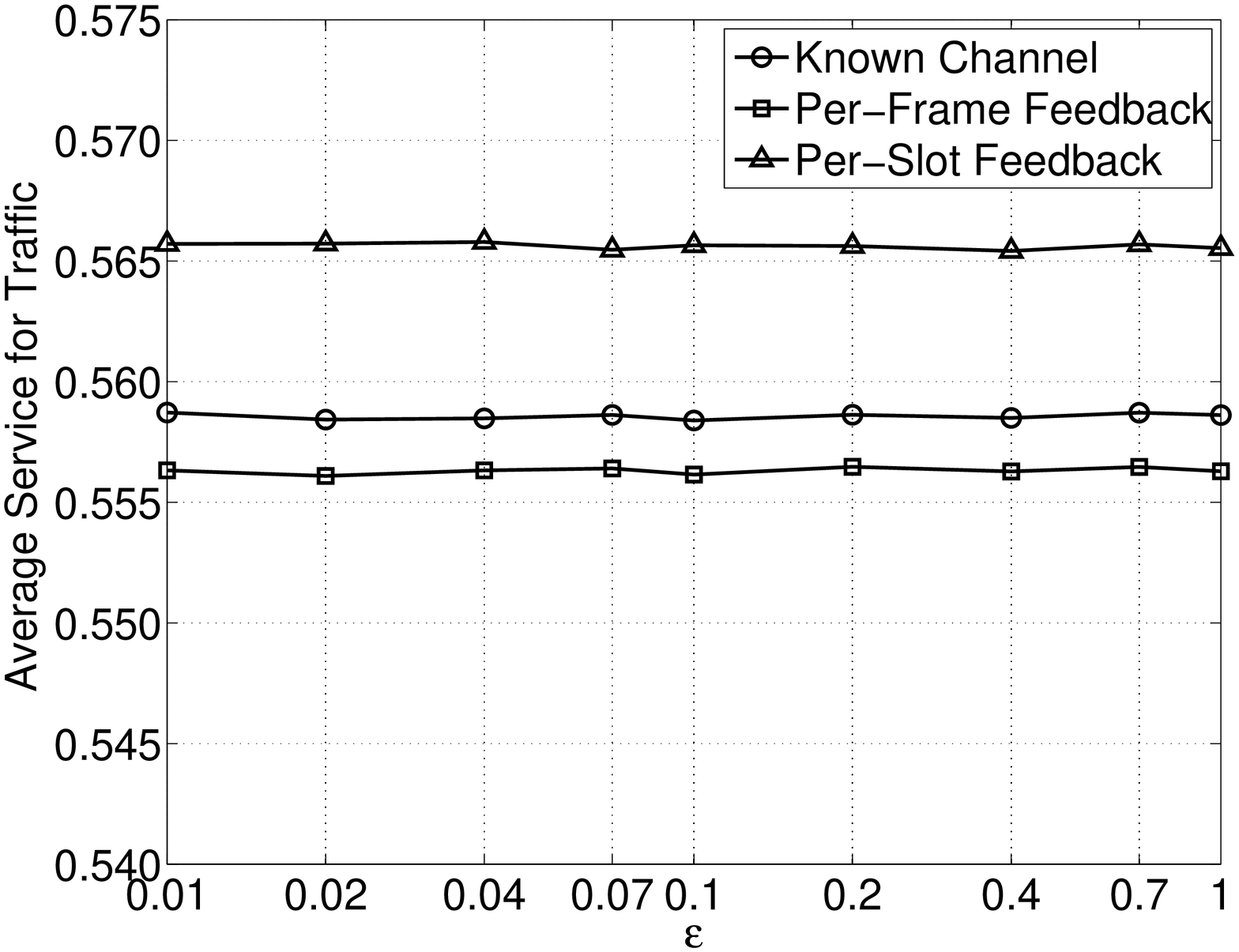}
	\caption{Average service when $w_l=6$}
	\label{rate_ten_w6_ch}
\end{figure}

\begin{figure}[t]
	\centering
	\includegraphics[angle=270, width=3.3in]{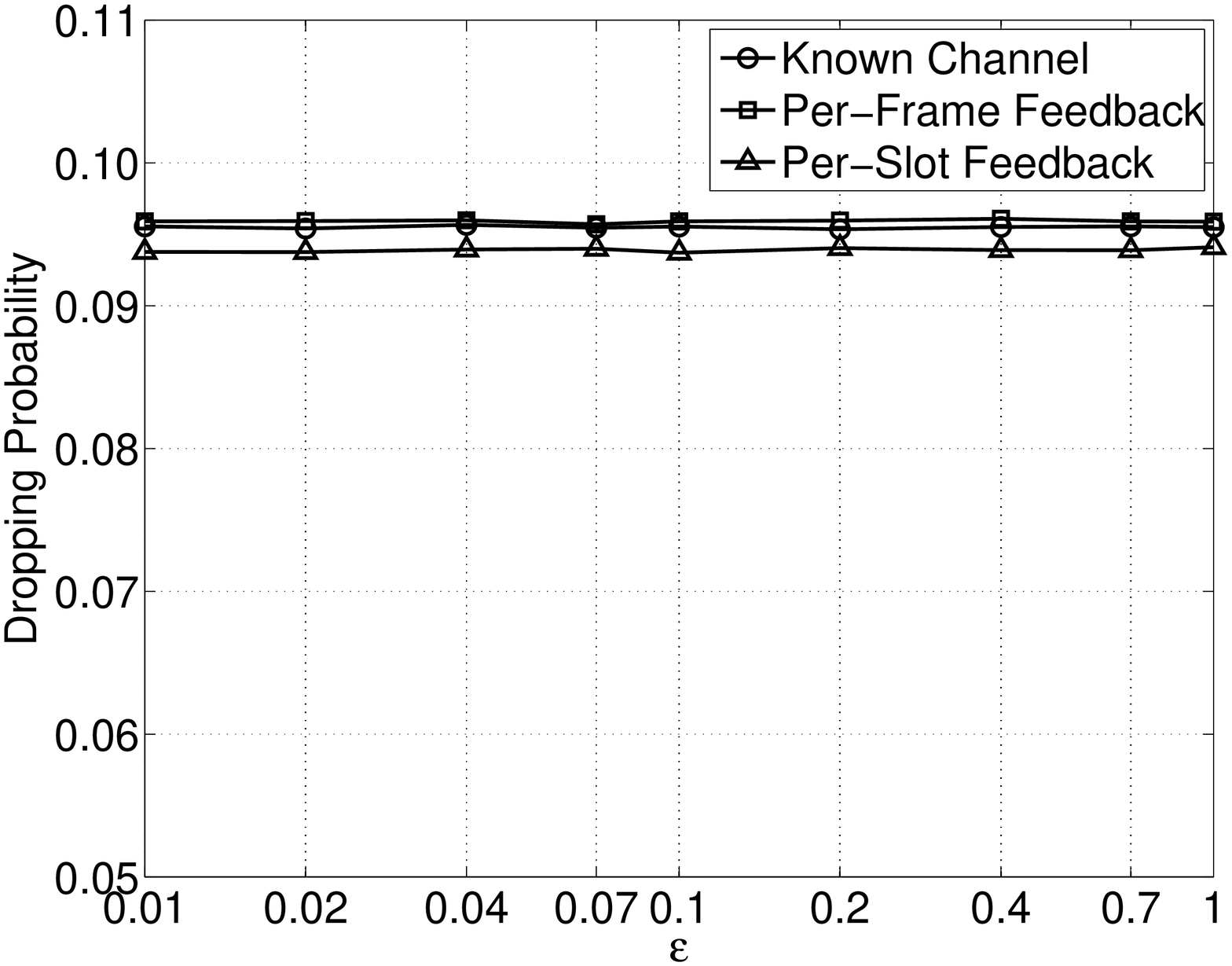}
	\caption{Dropping probability when $w_l=0$}
	\label{drop_prob_ten_w0_ch}
\end{figure}
\begin{figure}[t]
	\centering
	\includegraphics[angle=270, width=3.3in]{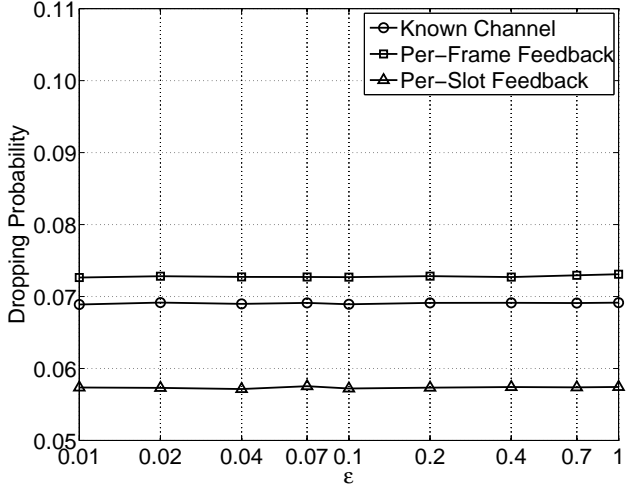}
	\caption{Dropping probability when $w_l=6$}
	\label{drop_prob_ten_w6_ch}
\end{figure}

%
%
\section{Conclusions}
\label{conclusions}

In this work we have presented an optimization formulation for the problem of scheduling real-time traffic in ad hoc networks under maximum delay constraints. The model allows for general arrival models with heterogeneous delay constraints. Using duality theory and a decomposition approach, we presented an optimal scheduler and proved that it fairly allocates data rates to all links and guarantees that the delay requirements are met at every flow. We further studied the impact of feedback at every time slot on the complexity of the optimal algorithm, and showed that for a certain simple scenario a greedy strategy can achieve the optimal solution with low complexity, recovering the results of \cite{Hou09a} for access-point networks.

%
%
\appendices
\section{Proof of Lemma \ref{expected_drift_kc}}

Before we prove Lemma \ref{expected_drift_kc}, we present the following fact.

\begin{fact}
\label{fact_online_opt_kc}
The optimization in (\ref{online_opt_sch_kc}) can be performed over $\mathcal{S}(a(k),c(k))_\mathcal{CH}$, the convex hull of $\mathcal{S}(a(k),c(k))$; that is,
\begin{align*}
	&\mathop{\arg\max}\limits_{s \in \mathcal{S}( a(k),c(k) )} \sum\limits_{l \in \mathcal{L}} [ \frac{1}{\epsilon} w_l+d_{l}(k) ] \sum_{ t \in \mathcal{T} } s_{lt} = \\
	& \mathop{\arg\max}\limits_{s \in \mathcal{S}( a(k),c(k) )_\mathcal{CH}} \sum\limits_{l \in \mathcal{L}} [ \frac{1}{\epsilon} w_l+d_{l}(k) ] \sum_{ t \in \mathcal{T} } s_{lt}.
\end{align*}
The reason for this comes from the fact that the objective function is linear and therefore there must be an optimal point $\tilde{s}^*( a(k),c(k),d(k) ) \in \mathcal{S}(a(k),c(k))$.
$\hfill \diamond$
\end{fact}

\begin{proof}[Proof of Lemma \ref{expected_drift_kc}]
\begin{align}
	E & \left[ V(d(k+1)) | d(k)=d \right] - V(d) \notag \\
 	= & E \left[ \frac{1}{2} \sum_{l \in \mathcal{L}} \{ [ d_l + \tilde{a}_{l}(k) - I_{l}^*(a(k),c(k),d) ]^+ \}^2 \right] - \sum_{l \in \mathcal{L}} \frac{d_l^2}{2} \notag \\
	\leq & E \left[ \frac{1}{2} \sum_{l \in \mathcal{L}} [ d_l + \tilde{a}_{l}(k) - I_{l}^*(a(k),c(k),d) ]^2 \right] - \sum_{l \in \mathcal{L}} \frac{d_l^2}{2} \notag \\
 	= & E \left[ \sum_{l \in \mathcal{L}} d_l [ \tilde{a}_{l}(k) - I_{l}^*(a(k),c(k),d) ] \right. \notag \\
 	& \left. + \frac{1}{2} \sum_{l \in \mathcal{L}} [ \tilde{a}_{l}(k) - I_{l}^*(a(k),c(k),d) ]^2 \right] \notag \\
 	\leq & E \left[ \sum_{l \in \mathcal{L}} d_l \tilde{a}_{l}(k) - d_l I_{l}^*(a(k),c(k),d) \right. \notag \\
 	& \left. + \frac{1}{2} \sum_{l \in \mathcal{L}} \tilde{a}_{l}^2(k) + a_{l}^2(k) \right] \label{def_I_l} \\
 	\leq & B_1 + \sum_{l \in \mathcal{L}} d_l \lambda_{l} (1-p_l) \notag \\
 	& - E \left[ \sum_{l \in \mathcal{L}} \left( \frac{1}{\epsilon}w_l + d_l \right) I_{l}^*(a(k),c(k),d) \right. \notag \\
 	& \left. - \sum_{l \in \mathcal{L}} \frac{1}{\epsilon}w_l I_{l}^*(a(k),c(k),d) \right] \notag \\
 	\leq & B_1 + \sum_{l \in \mathcal{L}} d_l \mu_{l}(\Delta) \label{def_mu_delta} \\
 	& - E \left[ \sum_{l \in \mathcal{L}} \left( \frac{1}{\epsilon}w_l + d_l \right) I_{l}^*(a(k),c(k),d) \right. \notag \\
 	& \left. - \sum_{l \in \mathcal{L}} \frac{1}{\epsilon}w_l I_{l}^*(a(k),c(k),d) \right] \notag
\end{align}
where (\ref{def_I_l}) follows from the definition of $I_{l}^*(a(k),c(k),d)$, (\ref{def_mu_delta}) follows from (\ref{inelastic_feasibility}), and
\begin{equation*}
	B_1 = \frac{1}{2} \sum_{l \in \mathcal{L}} ( \lambda_{l}^2 + \sigma_{l}^2 ) [ 1+(1-p_l)^2 ] + \lambda_{l} p_l (1-p_l).
\end{equation*}

Given the definition of $\mathcal{C}$, we have that $\mu(\Delta) \in \mathcal{C}/(1+\Delta)$ implies that there exist $\bar{\mu}(a,c) \in \mathcal{C}(a,c)$ for all $a$, $c$ and $(1+\Delta)\mu_{l}(\Delta) = E[ \bar{\mu}_{l}(a,c) ]$ for all $l \in \mathcal{L}$. For the rest of the proof we define $\bar{\mu}(a(k),c(k))$ to be such set of values associated to $(1+\Delta)\mu(\Delta)$. Hence:
\begin{align}
	E & \left[ V(d(k+1)) | d(k)=d \right] - V(d) \notag \\
 	\leq & B_1 + \sum_{l \in \mathcal{L}} d_l \mu_{l} (\Delta) \notag \\
 	& - E \left[ \sum_{l \in \mathcal{L}} \left( \frac{1}{\epsilon}w_l + d_l \right) I_{l}^*(a(k),c(k),d) \right] \notag \\
 	& + \sum_{l \in \mathcal{L}} \frac{1}{\epsilon}w_l E \left[ I_{l}^*(a(k),c(k),d) \right] \notag \\
 	\leq & B_1 + \sum_{l \in \mathcal{L}} d_l \mu_{l}(\Delta) \notag \\
 	& - E \left[ \sum_{l \in \mathcal{L}} \left( \frac{1}{\epsilon}w_l + d_l \right) \bar{\mu}_{l}(a(k),c(k)) \right] \label{fact_consequence} \\
 	& + \sum_{l \in \mathcal{L}} \frac{1}{\epsilon}w_l E \left[ I_{l}^*(a(k),c(k),d) \right] \notag \\
 	= & B_1 + \sum_{l \in \mathcal{L}} d_l \mu_{l}(\Delta) - \sum_{l \in \mathcal{L}} \left( \frac{1}{\epsilon}w_l + d_l \right) (1+\Delta)\mu_{l}(\Delta) \notag \\
 	& + \sum_{l \in \mathcal{L}} \frac{1}{\epsilon}w_l E \left[ I_{l}^*(a(k),c(k),d) \right] \notag \\
 	= & B_1 - \Delta \sum_{l \in \mathcal{L}} d_l \mu_{l}(\Delta) \notag \\
 	& - \frac{1}{\epsilon} \sum_{l \in \mathcal{L}} \left\{ w_l (1+\Delta) \mu_{l}(\Delta) - w_l E \left[ I_{l}^*(a(k),c(k),d) \right] \right\}, \notag
\end{align}
where \eqref{fact_consequence} follows from the fact that $\bar{\mu}(a(k),c(k)) \in \mathcal{C}(a(k),c(k))$ and Fact \ref{fact_online_opt_kc}. Therefore:
\begin{align*}
	E & \left[ V(d(k+1)) | d(k)=d \right] - V(d) \leq B_1 - B_2 \sum_{l \in \mathcal{L}} d_l \\
 	& - \frac{1}{\epsilon} \sum_{l \in \mathcal{L}} \left\{ w_l (1+\Delta) \mu_{l}(\Delta) - w_l E \left[ I_{l}^*(a(k),c(k),d) \right] \right\} \\
\end{align*}
where
\begin{equation*}
	B_2 = \Delta \min_{l \in \mathcal{L}} \left\{ \mu_{l}(\Delta) \right\}.
\end{equation*}
\end{proof}

\section{Proof of Theorem \ref{optimality_online_kc}}

From Lemma \ref{bound_expected_drift_kc} we know that
\begin{align*}
	\frac{1}{\epsilon} & \sum_{l \in \mathcal{L}}  \left\{ w_l \mu^*_{l} - w_l E \left[ I_{l}^*(a(k),c(k),d) \right] \right\} \\
 	\leq & B_1 - B_5 \sum_{l \in \mathcal{L}} d_l + V(d) - E \left[ V(d(k+1)) | d(k)=d \right] \\
 	\leq & B_1 + V(d) - E \left[ V(d(k+1)) | d(k)=d \right]
\end{align*}
since $B_5 \sum_{l \in \mathcal{L}} d_l \geq 0$. Taking expectations:
\begin{align*}
	\frac{1}{\epsilon} & E \left[ \sum_{l \in \mathcal{L}} \left\{ w_l \mu^*_{l} - w_l I_{l}^*(a(k),c(k),d(k)) \right\} \right] \\
	\leq & B_1 - E\left[ V(d(k+1)) \right] + E\left[ V(d(k)) \right].
\end{align*}

Adding the terms for $k=\{ 1, \ldots, K \}$ and dividing by $K$ we get:
\begin{align}
	\frac{1}{\epsilon} & E \left[ \sum_{l \in \mathcal{L}} w_l \mu^*_{l} - w_l \frac{1}{K} \sum_{k=1}^K I_{l}^*(a(k),c(k),d(k)) \right] \notag \\
	& \leq B_1 - \frac{E\left[ V(d(K+1)) \right]}{K} + \frac{E\left[ V(d(1)) \right]}{K} \notag \\
	& \leq B_1 + \frac{E\left[ V(d(1)) \right]}{K} \label{positive_mean_kc}
\end{align}
where (\ref{positive_mean_kc}) follows from the fact that the Lyapunov function V is non-negative. Assuming $E\left[ V(d(1)) \right] < \infty$ we get the following limit expression:
\begin{align*}
	\limsup_{K \rightarrow \infty} & E \left[ \sum_{l \in \mathcal{L}} w_l \mu^*_{l} - w_l \frac{1}{K} \sum_{k=1}^K I_{l}^*(a(k),c(k),d(k)) \right] \leq B \epsilon
\end{align*}
where $B=B_1$.\hfill $\blacksquare$

\section{Proof of Theorem \ref{optimality_greedy}}

The proof will use dynamic programming arguments: we will first note that if there is only one time slot remaining, the optimal decision is to schedule the link with the highest weight among the links that have a packet that remains to be transmitted, and then using induction we will prove that the best decision in any time slot is to schedule a backlogged link with the highest weight.

For simplicity in notation, define $\pi_l = \frac{1}{\epsilon} w_l + d_l $ for all $l \in \mathcal{L}$. Also define
\begin{equation*}
	U_{\rho}(\mathcal{L}, j) = \sum\limits_{l \in \mathcal{L}} \sum\limits_{c} \sum_{ t=T-j+1 }^{T} \pi_l c_{lt} s_{lt}(\rho, a, c) Pr(c)
\end{equation*}
to be the expected utility of policy $\rho$ when there are $j$ time slots remaining and the set of links that have a packet that needs to be transmitted is given by $\mathcal{L}$. Furthermore, define
\begin{equation*}
	U_{\rho}(\mathcal{L}, 0) = 0.
\end{equation*}
Finally, we will denote the greedy policy by $g$.

If there is only one time slot remaining, the optimal decision is to schedule one of the links that has the largest weight $\pi_l \bar{c}_l$ among the links that have a waiting packet, since this maximizes the expected utility. So the optimal decision	 is to use the greedy scheduler in the last time slot. Using induction, we assume that when there are $j$ time slots remaining, it is optimal to use the greedy scheduler. We will prove that if there are $j+1$ time slots remaining then it is also optimal to use the greedy scheduler.

When there are $j+1$ time slots, we need to determine which link to schedule in the first slot, and then use the greedy scheduler for the remaining $j$ time slots. Assume that the set of links that have a packet waiting to be transmitted is given by $\mathcal{L}$. If we schedule the link with the largest weight $\pi_l \bar{c}_l$, then the expected utility is given by
\begin{equation*}
	U_{g}(\mathcal{L}, j+1) = \pi_{ l^* } \bar{c}_{ l^* } + (1-\bar{c}_{ l^* }) U_g(\mathcal{L}, j) + \bar{c}_{ l^* } U_g(\mathcal{L} \setminus \{ l^* \}, j)
\end{equation*}
where
\begin{equation*}
	l^* \in \mathop{ \arg\max }_{ l \in \mathcal{L} } \pi_l \bar{c}_l.
\end{equation*}
If we decide to schedule link
\begin{equation*}
	\tilde{l} \notin \mathop{ \arg\max }_{ l \in \mathcal{L} } \pi_l \bar{c}_l,
\end{equation*}
then the expected utility is given by
\begin{align*}
	U_{\tilde{\rho}}(\mathcal{L}, j+1) = & \pi_{ \tilde{l} } \bar{c}_{ \tilde{l} } + (1-\bar{c}_{ \tilde{l} }) U_g(\mathcal{L}, j) + \bar{c}_{ \tilde{l} } U_g(\mathcal{L} \setminus \{ \tilde{l} \}, j)\\
	= & \pi_{ \tilde{l} } \bar{c}_{ \tilde{l} } + \pi_{ l^* } \bar{c}_{ l^* } + ( 1-\bar{c}_{l^*} ) ( 1-\bar{c}_{\tilde{l}} ) U_g(\mathcal{L}, j-1) \\
	   & + ( 1-\bar{c}_{l^*} ) \bar{c}_{\tilde{l}} U_g(\mathcal{L} \setminus \{ \tilde{l} \}, j-1) \\
	   & + \bar{c}_{l^*} ( 1-\bar{c}_{\tilde{l}} ) U_g(\mathcal{L} \setminus \{ l^* \}, j-1) \\
	   & + \bar{c}_{l^*} \bar{c}_{\tilde{l}} U_g(\mathcal{L}  \setminus \{ l^*, \tilde{l} \}, j-1) \\
	= & \pi_{ l^* } \bar{c}_{ l^* } + ( 1-\bar{c}_{l^*} ) \left[ \pi_{ \tilde{l} } \bar{c}_{ \tilde{l} } + ( 1-\bar{c}_{\tilde{l}} ) U_g(\mathcal{L}, j-1) \right. \\
	   & \left. + \bar{c}_{\tilde{l}} U_g(\mathcal{L} \setminus \{ \tilde{l} \}, j-1) \right] \\
	   & + \bar{c}_{l^*} \left[ \pi_{ \tilde{l} } \bar{c}_{ \tilde{l} } + ( 1-\bar{c}_{\tilde{l}} ) U_g(\mathcal{L} \setminus \{ l^* \}, j-1) \right. \\
	   & \left. + \bar{c}_{\tilde{l}} U_g(\mathcal{L} \setminus \{ l^* , \tilde{l} \}, j-1) \right]. \\
\end{align*}

So in order to prove that $U_{g}(\mathcal{L}, j+1) \geq U_{\tilde{\rho}}(\mathcal{L}, j+1)$ it suffices to show that
\begin{align}
\label{condition1_uc_ps}
	U_g(\mathcal{L}, j) \geq & \pi_{ \tilde{l} } \bar{c}_{ \tilde{l} } + ( 1-\bar{c}_{\tilde{l}} ) U_g(\mathcal{L}, j-1) \\
	& + \bar{c}_{\tilde{l}} U_g(\mathcal{L} \setminus \{ \tilde{l} \}, j-1) \notag
\end{align}
and
\begin{align}
\label{condition2_uc_ps}
	U_g(\mathcal{L} \setminus \{ l^* \}, j) \geq & \pi_{ \tilde{l} } \bar{c}_{ \tilde{l} } + ( 1-\bar{c}_{\tilde{l}} ) U_g(\mathcal{L} \setminus \{ l^* \}, j-1) \\
	& + \bar{c}_{\tilde{l}} U_g(\mathcal{L} \setminus \{ l^* , \tilde{l} \}, j-1). \notag
\end{align}
From the assumption that the greedy scheduler is optimal when there are $j$ slots remaining, it is clear that \eqref{condition1_uc_ps} and \eqref{condition2_uc_ps} are true, so the greedy scheduler is indeed optimal when there are $j+1$ slots remaining.
\hfill $\blacksquare$

%
%
%
%
%


%

\bibliographystyle{IEEEtran}
\bibliography{IEEEabrv,../bibtex/delay}

\end{document}